\documentclass[10pt,twocolumn,twoside]{IEEEtran}
\IEEEoverridecommandlockouts

\usepackage{amsmath,amsfonts,amssymb}
\usepackage{algorithm}
\usepackage{algpseudocode}
\usepackage{array}
\usepackage[caption=false,font=normalsize,labelfont=sf,textfont=sf]{subfig}
\usepackage[table]{xcolor}
\usepackage{textcomp}
\usepackage{stfloats}
\usepackage{url}
\usepackage{verbatim}
\usepackage{graphicx}
\usepackage{cite}
\usepackage{mathtools}
\usepackage{xcolor}
\usepackage{amsthm}
\usepackage{booktabs}
\usepackage[]{hyperref} 
\usepackage{multicol}
\usepackage{comment}
\usepackage{pifont}

\theoremstyle{definition}
\newtheorem{definition}{Definition}[section]
\newtheorem{theorem}{Theorem}

\newtheorem{property}{Property}
\newtheorem{prop}{Proposition}
\newtheorem*{remark}{Remark}

\DeclarePairedDelimiter\floor{\lfloor}{\rfloor}

\def\BibTeX{{\rm B\kern-.05em{\sc i\kern-.025em b}\kern-.08em
    T\kern-.1667em\lower.7ex\hbox{E}\kern-.125emX}}

\usepackage{lipsum}

\title{\LARGE \bf 
Making Aggregations Reliable: Realizability Guarantees for Battery Fleets with Heterogeneous Power and Energy Limits
}

\author{Mazen Elsaadany and Mads R. Almassalkhi
\vspace{-0.5cm}
}

\begin{document}
\begingroup
\allowdisplaybreaks

\maketitle

\begin{abstract}
Aggregated battery energy storage systems (BESS) enable large fleets of heterogeneous battery elements to participate in system-level optimization and electricity markets. Scheduling each element independently is computationally impractical at scale. While many aggregate battery models rely on convex relaxations, they often ignore element complementarity constraints, leading to dispatch solutions that may be infeasible when implemented on individual battery elements. This paper develops a realizable composite battery model for parameter-heterogeneous BESS fleets that guarantees feasibility at the element-level while preserving computational tractability. We derive simple linear conditions under which aggregate charging and discharging trajectories can be safely disaggregated while respecting individual power limits, energy limits, and complementarity constraints under a priority-based controller. Numerical experiments in a unit-commitment setting demonstrate that the proposed realizable composite battery formulation produces feasible dispatch solutions. Solve times are effectively independent of system size, unlike micro-model mixed-integer formulations. Solutions obtained from the proposed formulation converge to the optimal benchmark as control granularity is refined. Additional studies illustrate the robustness of the framework to moderate violations of key modeling assumptions, including heterogeneous power-to-energy ratios.

\end{abstract}



\section{Introduction}


FERC Order No.~2222 enabled aggregated battery energy storage systems (BESS) fleets to participate in wholesale markets and provide energy and ancillary services~\cite{FERC2222}. However, explicitly scheduling and dispatching each element in a large fleet is computationally impractical~\cite{MultiBatteryEV2023}. This motivates aggregating and coordinating geographically distributed batteries as virtual power plants (VPPs) that reduce problem complexity while still producing dispatch decisions that can be physically delivered by the underlying fleet.

Most aggregate modeling approaches characterize each battery element's flexibility set and define the fleet-level flexibility as the Minkowski sum of these individual sets\cite{Muller2019AggDisagg,Wen2022ExactAFR,MultiBatteryEV2023,Ozturk2024VertexBased}. Since exact Minkowski sums are intractable at scale, the literature focuses on tractable inner or outer approximations. Examples include zonotope-based approximations, which scale well but can be conservative for heterogeneous populations~\cite{Muller2015Zonotopes,Muller2019AggDisagg}, homothetic and union-based constructions that improve accuracy by combining affine images of a base set~\cite{Nazir2018UnionBasedMinkowski}, and methods that derive the exact aggregate feasible set but require simplifications for long time horizons~\cite{Wen2022ExactAFR}. Recent vertex-based methods construct compact inner approximations from a structured subset of the vertices of the Minkowski sum~\cite{Ozturk2024VertexBased}. Finally, multi-battery models represent heterogeneous fleets using a small number of representative devices to trade off accuracy and complexity~\cite{MultiBatteryEV2023}.

A key challenge is that most aggregate models rely on convex representations and therefore sidestep the non-convex complementarity constraint that prevents simultaneous charging and discharging at the element level. Ignoring this constraint can produce schedules that appear feasible at the aggregate level but are infeasible when implemented on individual elements~\cite{prat_2024,RobustNawaf}. In particular, convex relaxations can underestimate state of energy (SOE), leading to dispatches that cannot be realized without violating element constraints~\cite{RobustNawaf}.

Existing approaches handle complementarity in several ways. Mixed-integer formulations enforce it exactly but scale poorly with fleet size and horizon length~\cite{MILP_for_CC,prat_2024}. Other work derives conditions under which convex relaxations don't violate complementarity, typically under strong assumptions on costs, prices, or parameter values~\cite{Sufficient_Cond_for_Relax,Lin_Exact_Cond_Ex_Ante,positive_prices,Jakob_Pos_Quadratic_Cost,prat_2024}. Penalty-based formulations discourage simultaneous charging and discharging by assigning a penalty term in the objective; however, the resulting guarantees are typically tied to a specific application (e.g., DCOPF) and do not readily generalize to other optimization settings~\cite{garifi_convex_2020}. Robust convex restrictions can guarantee feasibility by bounding worst-case energy evolution, often at the expense of increased conservatism over longer horizons~\cite{RobustNawaf}.

In summary, there is a lack of aggregate battery models that both (i) scale to large heterogeneous fleets and (ii) guarantee realizable dispatch in the presence of charging--discharging complementarity constraints. Other work addressed this challenge for homogeneous fleets by introducing a linear composite model with a priority-based disaggregation strategy that guarantees feasibility at the element level while preserving convexity~\cite{Composite_Homogeneous}. Extending such guarantees to heterogeneous fleets remains an open problem.

This paper addresses this gap by developing a realizable aggregate modeling framework for heterogeneous BESS fleets. The contributions are:
\begin{enumerate}
    \item We extend the composite battery framework to heterogeneous elements and derive a linear aggregate model that captures heterogeneous power and energy limits.
    \item We generalize the priority-based disaggregation strategy and provide sufficient conditions under which any feasible aggregate dispatch is realizable while respecting element-level power limits, energy limits, and complementarity.
    \item We characterize the resulting aggregate feasible region and show how accuracy improves with finer control granularity while maintaining the computational tractability needed for large-scale optimization.
\end{enumerate}

The remainder of the paper is organized as follows. Section~II presents the element-level model and the composite battery formulation. Section~III introduces the priority stack controller and derives sufficient conditions for realizable dispatch. Section~IV evaluates the proposed model using numerical case studies, including scalability and the effects of relaxing key assumptions. Section~V concludes and outlines future research directions.
\section{Composite Battery Modeling}
To enable system-level optimization without explicitly modeling every battery element, we introduce a composite battery model that aggregates the power and energy states of a fleet of $N$ independently controllable battery elements. We first define the element-level model, then derive the corresponding composite model and discuss the challenges that arise when mapping aggregate charging/discharging power to individual battery elements.
\subsection{Battery Elements}
Each element $i\in\{1,\ldots,N\}$ is operated in discrete \emph{control time steps} indexed by $l$, with duration $\delta t$. The charging and discharging powers are denoted by $P^i_\text{c}(l)$ and $P^i_\text{d}(l)$, with efficiencies $\eta^i_\text{c}$ and $\eta^i_\text{d}$, respectively. The SOE of element $i$ at time step $l$ is $E^i(l)$. Each element has power limits $P^i_{\text{c,max}}$, $P^i_{\text{d,max}}$, energy limit $E^i_{\max}$, and an initial SOE $E^i_0$. The element constraints are summarized by
\begin{subequations}
\label{eqn:Elements Constraints}
    \begin{align}
        &E^i(l+1)=E^i(l)+ \delta t\eta^i_\text{c}P_\text{c}^i(l)
        -\delta t\frac{1}{\eta^i_\text{d}}P_\text{d}^i(l) \label{eqn:Elements SOC} \\
        &E^i(0) = E^i_0 \label{eqn:Elements SOC IC}\\
        &0\leq E^i(l+1)\leq {E}_\text{max}^i\label{eqn:Elements SOC Limit} \\
        &0\leq P_\text{c}^i(l)\leq {P}_\text{c,max}^i\label{eqn:Elements Pc Limit}\\
        &0\leq P_\text{d}^i(l)\leq {P}_\text{d,max}^i\label{eqn:Elements Pd Limit}\\
        &0=P_\text{c}^i(l)P_\text{d}^i(l), \label{eqn:Elements CC}
    \end{align}
\end{subequations}
for all $i=1,\ldots,N$ and $l=0,\ldots,L-1$.
\subsection{Composite Battery}
The composite battery aggregates the $N$ elements and is operated on coarser \emph{scheduling time steps} of duration $\Delta t=M\delta t$, where $M\in\mathbb{N}^+$. We index scheduling steps by $[k]$, where $k=0,\ldots,K-1$ and $L=MK$. Let $P_\text{c}[k]$ and $P_\text{d}[k]$ denote the composite charging and discharging powers. A piecewise-constant composite schedule $\{P_\text{c}[k],P_\text{d}[k]\}_{k=0}^{K-1}$ is realized by element-level actions if
\begin{align}
     \sum_{i=1}^{N} P_\text{c}^i(l) =P_\text{c}[k] \quad  \land \quad 
     \sum_{i=1}^{N} P_\text{d}^i(l) = P_\text{d}[k],
    \label{eqn:Sum of Element}
\end{align}
for all $l$ with $k=\floor{\tfrac{l}{M}}$. The composite SOE is defined as the
sum of element SOEs at the start of each scheduling interval:
\begin{align}\label{eqn:SOE mapping}
    E[k]=\sum_{i=1}^{N} E^i(Mk).
\end{align}

Under the assumption of homogeneous charging and discharging efficiencies
($\eta_\mathrm{c}^i=\eta_\mathrm{c}$ and $\eta_\mathrm{d}^i=\eta_\mathrm{d}$),
the composite SOE evolves as
\begin{subequations}
    \label{eqn:Composite SOE}
    \begin{align}
    E[k+1]&=E[k]+ \Delta t\eta_\text{c}P_\text{c}[k] -\Delta t\frac{1}{\eta_\text{d}}P_\text{d}[k] \label{eqn:Composite SOE Update} \\
    E[0] &= \sum_{i}E^i_0. \label{eqn:Composite SOE IC}
\end{align}
\end{subequations}
according to the following dynamics, which follow directly from \eqref{eqn:Elements SOC},\eqref{eqn:Elements SOC IC} and \eqref{eqn:SOE mapping}.

Similarly, the composite power and energy limits follow from \eqref{eqn:Elements SOC Limit}-\eqref{eqn:Elements Pd Limit}, \eqref{eqn:Sum of Element} and \eqref{eqn:SOE mapping}:
\begin{subequations}\label{eqn:Composite Naive Constraints}
\begin{align}
       0\le &P_\text{c}[k]\le \sum_i P_\text{c,max}^i\\
    0\le &P_\text{d}[k]\le \sum_i P_\text{d,max}^i\\
    0\le & E[k]\le \sum_i E_\text{max}^i.
\end{align}
\end{subequations}

While \eqref{eqn:Composite SOE}--\eqref{eqn:Composite Naive Constraints} capture aggregate power and energy behavior, they do not ensure that element-level constraints are satisfied when the composite power trajectory is disaggregated to the individual elements. In particular, the composite model does not enforce element complementarity or prevent individual elements from violating their energy limits. To address this limitation, we introduce the notion of \emph{realizability}, which formalizes when a composite dispatch can be implemented by the underlying battery elements. Next, we derive sufficient conditions on the composite dispatch that guarantee realizability under heterogeneous battery parameters.

\begin{definition}[Realizable Composite Dispatch]
A composite dispatch sequence $\{P_\mathrm{c}[k], P_\mathrm{d}[k]\}_{k=0}^{K-1}$ is \emph{realizable} if there exist element-level trajectories $\{P_\mathrm{c}^i(l), P_\mathrm{d}^i(l)\}_{l=0}^{L-1}~\forall i=1,\ldots,N$ that satisfy the aggregation constraints in \eqref{eqn:Sum of Element} and the element constraints in \eqref{eqn:Elements Constraints} for all $i$ and $l$.
\end{definition}

The realizability of a composite dispatch depends on how composite power trajectories are disaggregated to individual battery elements. In the next section, we introduce a priority-based disaggregation policy and derive sufficient conditions under which a composite dispatch is guaranteed to be realizable for heterogeneous battery fleets.
\begin{figure}
    \centering
    \includegraphics[width=0.95\linewidth]{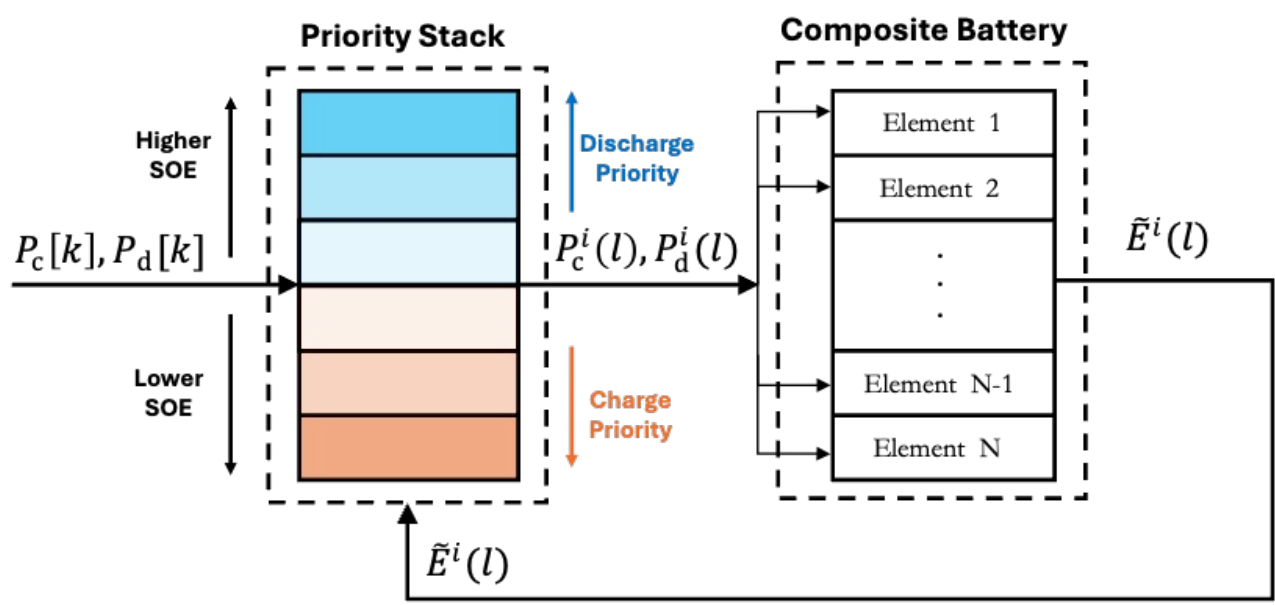}
    \caption{The composite battery, its elements, and the priority stack controller.}
    \label{fig:PSC}
\end{figure}
\section{Composite Trajectory Disaggregation}\label{sec: Composite Disagg}

This section introduces a priority-based disaggregation policy that maps composite charging and discharging power trajectories $\{P_\mathrm{c}[k], P_\mathrm{d}[k]\}_{k=0}^{K-1}$ to element-level commands $\{P_\mathrm{c}^i(l), P_\mathrm{d}^i(l)\}_{l=0}^{L-1}$ for each battery element $i=1,\ldots,N$. The proposed controller operates at the control time scale $l$, assigning charging and discharging priorities based on the elements’ SOEs, as illustrated in~Fig.~\ref{fig:PSC}.

Using this disaggregation policy, we derive sufficient conditions on the composite power and energy trajectories under which the resulting element-level dispatch is guaranteed to satisfy all element power limits, energy limits, and complementarity constraints.

\subsection{Priority Stack Controller (PSC)}
For every control time step $l$, the priority stack controller (PSC) sorts all $N$ battery elements according to their \emph{normalized state of energy (SOE)}. For each element $i$, the normalized SOE is defined as

\begin{align}\label{eqn:Element Normalized SOE}
    \tilde{E}^i(l) \triangleq \frac{E^i(l)}{E^i_{\max}},
\end{align}
which represents the fraction of usable energy stored in the element. Enforcing $\tilde{E}^i(l)\in[0,1]$ is equivalent to enforcing the element energy limits $0\le E^i(l)\le E^i_{\max}$.

Elements with lower normalized SOE are assigned higher charging priority, while elements with higher normalized SOE are assigned higher discharging priority. Similar priority-based control schemes have been proposed in the literature~\cite{karan_kalsi,equal_unequal,equal_unequal2}.

Given a composite charging or discharging command at scheduling time step $k$, the PSC activates a subset of elements according to this priority order, resulting in the $N_\mathrm{c}(l)$  elements with the lowest normalized SOE charging and $N_\mathrm{d}(l)$ elements with the highest normalized SOE discharging at control time step $l$. The PSC algorithm is summarized in Algorithm~\ref{alg:PSC} and implies Property~\ref{property:PSC}.

\begin{algorithm}
\caption{Priority Stack Controller (PSC)}
\label{alg:PSC}
\begin{algorithmic}[1]
{\small
    \State Given $P_\mathrm{c}[k]$ and $P_\mathrm{d}[k]$
    \For{each control time step $l$}
        \State Compute normalized SOE: $\tilde{E}^i(l)\gets E^i(l)/E^i_{\max}$
        \State Sort elements by $\tilde{E}^i(l)$ in ascending order
        \State Initialize $P_\mathrm{c}^i(l)\gets 0$, $P_\mathrm{d}^i(l)\gets 0$ for all $i$
        
        \State \textbf{Charging:} for $i=1,2,\ldots,N$ (lowest-$\tilde{E}$ first)
        \If{$\sum_j P_\mathrm{c}^j(l) < P_\mathrm{c}[k]$}
            \State $P_\mathrm{c}^i(l)\gets \min\!\left\{P_{\mathrm{c},\max}^i,\; P_\mathrm{c}[k]-\sum_j P_\mathrm{c}^j(l)\right\}$
        \EndIf
        
        \State \textbf{Discharging:} for $i=N,N-1,\ldots,1$ (highest-$\tilde{E}$ first)
        \If{$\sum_j P_\mathrm{d}^j(l) < P_\mathrm{d}[k]$}
            \State $P_\mathrm{d}^i(l)\gets \min\!\left\{P_{\mathrm{d},\max}^i,\; P_\mathrm{d}[k]-\sum_j P_\mathrm{d}^j(l)\right\}$
        \EndIf
        
        \State Update SOE: $E^i(l)\rightarrow E^i(l+1)$ using \eqref{eqn:Elements SOC}
    \EndFor
}
\end{algorithmic}
\end{algorithm}

\begin{property}[Ordering]\label{property:PSC}
    Under PSC, if $\tilde{E}^i(l) \ge \tilde{E}^j(l)$,  then $P^i_\text{c}(l)>0 \implies P^j_\text{c}(l) = P^j_\text{c,max}$ and $P^j_\text{d}(l)>0 \implies P^i_\text{d}(l) = P^i_\text{d,max}$.
\end{property}


To derive sufficient conditions for realizable composite dispatch, we impose a set of structural assumptions on the battery elements. 
\subsection{Assumptions for Realizability Analysis}
\label{subsec:assumptions}
The sufficient conditions for realizable composite dispatch derived in the
following subsection are obtained under the following assumptions:
\begin{enumerate}
   \item \textbf{Symmetric power limits:} Each element has equal maximum charging and discharging power, i.e., $P_{\mathrm{c},\max}^i = P_{\mathrm{d},\max}^i = P_\text{max}^i$ for all $i$.

    \item \textbf{Uniform power-to-energy ratios:} The ratio of maximum power to
    energy capacity is identical across all elements, i.e.,
    \begin{align}
        \frac{P_{\mathrm{c},\max}^i}{E_\text{max}^i} = \alpha_\mathrm{c} \quad \forall i, \qquad
        \frac{P_{\mathrm{d},\max}^i}{E_\text{max}^i} = \alpha_\mathrm{d} \quad \forall i.
    \end{align}

    \item \textbf{Homogeneous efficiencies:} All battery elements share identical charging and discharging efficiencies, i.e., $\eta_\mathrm{c}^i = \eta_\mathrm{c}$ and
    $\eta_\mathrm{d}^i = \eta_\mathrm{d}$ $~\forall i$.
\end{enumerate}


\begin{remark}[Practical Considerations]
In practical heterogeneous battery fleets, Assumption~2 is likely the most restrictive assumption since it imposes a uniform power-to-energy ratio, which generally ranges from 0.25 (4-hour battery) to 1.0 (1-hour battery) in practice. By contrast, Assumption~1 is practical since many commercial battery systems are designed with symmetric inverter ratings~\cite{tesla_2019}. More general energy storage systems, such as pumped hydro or fuel cells with electrolyzers may violate Assumption~1. Similarly, Assumption~3 is deemed mild for batteries as charging/discharging efficiencies generally vary within a narrow band (e.g., $92\%\pm 3\%$). For this reason, Section~\ref{sec:relaxing_P_E_ratios} focuses on effects of relaxing Assumption~2.
\end{remark}
\subsection{Sufficient Conditions for Realizable Dispatch}\label{subsec:sufficient_conditions}
Under the assumptions above, we derive sufficient conditions on the composite power and energy trajectories $P_\text{c}[k]$ $P_\text{d}[k]$, and $E[k]$ that guarantee the composite dispatch can be realized by the battery fleet under the PSC. By design, the PSC never commands charging or discharging power beyond individual element limits, so the remaining challenges are ensuring that (i) element energy limits are not violated over time and (ii) complementarity holds at the element level.

\subsubsection{Guaranteeing Element Complementarity Constraints}
To satisfy element-level complementarity constraints $(P_\text{c}^i(l)P_\text{d}^i(l)=0,\forall i,\forall l)$, the PSC must not assign charging and discharging commands to the same battery element at any control time step.  Any overlap between the $N_\text{c}(l)$ charging elements and the $N_\text{d}(l)$ discharging elements would violate complementarity, which implies the condition
\begin{equation}\label{eqn:NcNd}
    N_\text{c}(l) + N_\text{d}(l)\leq N, \quad \forall l.
\end{equation}

\begin{prop}[Sufficient Conditions for Complementarity] \label{prop: cutting plane}
Under the PSC, if $P_\text{c}[k]$ and $P_\text{d}[k]$ satisfy 
\begin{equation}
\label{eqn:new cut plane}
    P_\text{c}[k]+P_\text{d}[k]\leq \left(\sum_{i=1}^{N}P^i_\text{max}\right)- \max_{\forall i}\{P^i_\text{max}\},
\end{equation}
then the resulting $N_\text{c}(l)$ and $N_\text{d}(l)$ satisfy \eqref{eqn:NcNd}.
\end{prop}
\begin{proof}
    See Appendix~\ref{app:prop_1_proof}.
\end{proof}

Proposition~\ref{prop: cutting plane} introduces an additional linear constraint \eqref{eqn:new cut plane} on the composite charging and discharging powers that guarantees element-level complementarity under the PSC. Fig.~\ref{fig:feasible_regions} illustrates the resulting composite feasible regions defined by \eqref{eqn:NcNd} and \eqref{eqn:new cut plane}.

\begin{figure}[ht]
\centering
\includegraphics[width=.475\columnwidth]{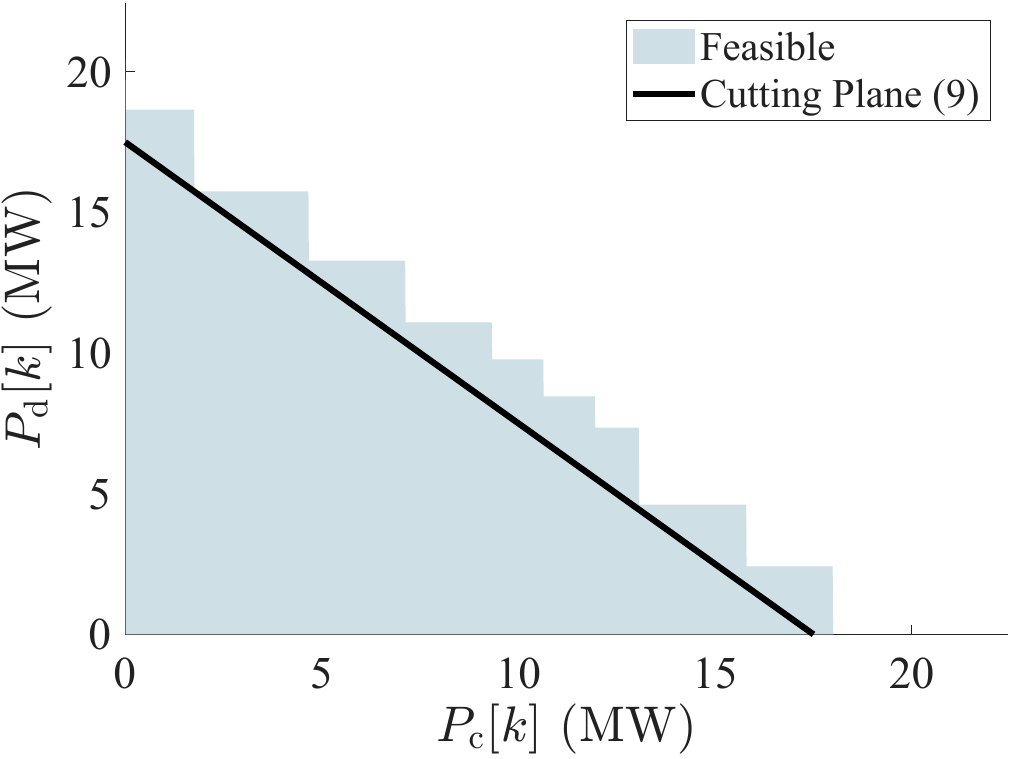}~
\includegraphics[width=.475\columnwidth]{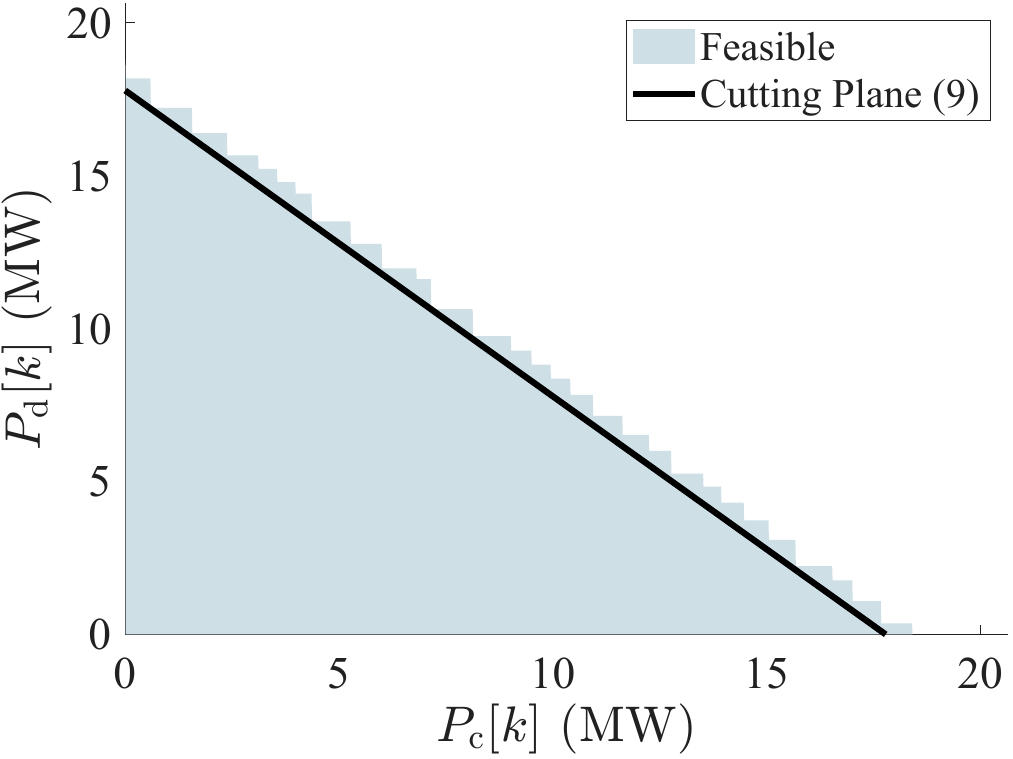}
\caption{Feasible region for $N=10$ and $P^i_\text{max}\in[1,3]$ MW (left) and $N=30$ and $P^i_\text{max}\in[0.33,1]$MW (right). SOE limits are omitted here.}
\label{fig:feasible_regions}
\end{figure}

\subsubsection{Guaranteeing Element Energy Limits}
Under the priority stack controller, battery elements are charged and discharged individually at each control time step, resulting in heterogeneous trajectories of normalized SOE $\tilde{E}^i(l)$ across the fleet. At the composite level, we define the normalized composite SOE as
\begin{align}\label{eqn:Composite Normailzed SOE}
    \tilde{E}(l)
    \triangleq
    \frac{\sum_i E^i(l)}{\sum_i E_\text{max}^i}
    =
    \frac{E(l)}{\sum_i E_\text{max}^i},
\end{align}
where $E(l)$ is the aggregate energy at control step $l$. While enforcing $\tilde{E}(l)\in[0,1]$ ensures that the composite battery remains within its aggregate energy limits, it does not guarantee that all individual elements satisfy $\tilde{E}^i(l)\in[0,1]$. In particular, under the PSC, normalized element SOEs can spread around the composite value.

From the definitions of normalized element and composite SOE, it follows that
\begin{align}\label{eqn:Norm SOE order}
    \tilde{E}^{\min}(l) \le \tilde{E}(l) \le \tilde{E}^{\max}(l),
\end{align}
where $\tilde{E}^{\min}(l)$ and $\tilde{E}^{\max}(l)$ denote the minimum and
maximum normalized element SOEs at time step $l$, respectively. This difference between element and composite normalized SOE can cause violations of element energy limits even with $\tilde{E}(l)\in[0,1]$.

To formalize this notion, we define the maximum normalized SOE spread as
\begin{align}
    \Delta \tilde{E}(l)
    \triangleq
    \max_{i,j}\left\{ \left| \tilde{E}^i(l)-\tilde{E}^j(l) \right| \right\}.
\end{align}

\begin{prop}[Bounding SOE spread]\label{prop: SOE difference bound}
Under the PSC,
\begin{align}
    \Delta \tilde{E}(l+1)
    \le
    \max\left\{ \Delta \tilde{E}(l), \tilde{\varepsilon} \right\},
    \quad \forall l,
\end{align}
where
\begin{equation}\label{eq:DEmax-def}
    \tilde{\varepsilon}
    :=
    \delta t\left(\eta_\text{c}\alpha_\text{c}
    + \frac{\alpha_\text{d}}{\eta_\text{d}}\right).
\end{equation}
\end{prop}
\begin{proof}
    See Appendix~\ref{app:proof_prop2}.
\end{proof}


\begin{prop}[Admissible element energy limits]\label{prop:SOC new limits}
Under the PSC, a composite dispatch that satisfyies
\begin{equation}\label{eqn:Composite_Energy_Limits}
    \tilde{\varepsilon} \le \tilde{E}[k] \le 1-\tilde{\varepsilon}
\end{equation}
is disaggregated into element-level trajectories that respect all individual element energy limits,
provided that (i) $\max_{i,j}\{|\tilde{E}^i_0-\tilde{E}^j_0|\}\le \tilde{\varepsilon}$ and (ii) the control time step $\delta t=\Delta t/M$ is sufficiently small to ensure
$\tilde{\varepsilon}\le \tfrac{1}{2}$.
\end{prop}
\begin{proof}
    See Appendix~\ref{app:proof_prop3}.
\end{proof}
 By maintaining a buffer of size $\tilde{\varepsilon}$ between
the composite SOE trajectory and the aggregate energy limits, sufficient
headroom is preserved to accommodate the maximum SOE spread induced by the PSC,
ensuring that no individual element reaches its energy limits.

\subsection{Composite Battery Model for Heterogeneous Fleets}
\label{subsec:composite_theorem}
We now summarize the proposed aggregate model for a heterogeneous battery fleet and state the conditions under which it admits realizable dispatch. The result parallels the composite battery formulation developed for homogeneous fleets in~\cite{Composite_Homogeneous}, while explicitly accounting for heterogeneity in power and energy ratings.

\begin{theorem}[Composite Battery Model for Heterogeneous Fleets] \label{thm:heterogeneous_composite}
Under Assumptions~1--3 in Section~\ref{subsec:assumptions}, any composite dispatch sequence $\{P_\mathrm{c}[k],P_\mathrm{d}[k]\}_{k=0}^{K-1}$ that satisfies

\begin{subequations} \label{eqn:Composite_Model_Theorem}
\begin{align}
    E[k+1]&=E[k]+ \Delta t\eta_\text{c}P_\text{c}[k] -\Delta t\frac{1}{\eta_\text{d}}P_\text{d}[k] \label{eqn:Thrm_Composite SOE Update} \\
    E[0] &= \sum_{i}E^i_0\label{eqn:Thrm_Composite SOE IC} \\
    0\le &P_\text{c}[k]\le \sum_i P_\text{max}^i \label{eqn:Thrm_Pc_limit}\\
    0\le &P_\text{d}[k]\le \sum_i P_\text{max}^i\label{eqn:Thrm_Pd_limit}\\
    P_\text{c}[k]&+P_\text{d}[k]\le \sum_i P_\text{max}^i -\max_i\{P_\text{max}^i\}\label{eqn:Thrm_new_cut_plane}\\ 
    \tilde{\varepsilon}\sum_i &E_\text{max}^i\le E[k]\le (1-\tilde{\varepsilon})\sum_i E_\text{max}^i \label{eqn:Thrm_energy}
\end{align}
\end{subequations}
where $\tilde{\varepsilon}$ is defined in \eqref{eq:DEmax-def}, is realizable
under the priority stack controller provided that $\max_{i,j}\{|\tilde{E}^i_0-\tilde{E}^j_0|\}\le \tilde{\varepsilon}$ and $\tilde{\varepsilon}<\frac{1}{2}$.
\end{theorem}

\begin{proof}
See Appendix~\ref{appx_Thrm_proof}.
\end{proof}


\begin{remark}[Practical Implementation Considerations]  The realizability guarantees derived in Section~III are sufficient and rely on structural assumptions: symmetric charge/discharge power limits, uniform power‑to‑energy ratios, and homogeneous efficiencies across elements. We also assume a small initial normalized SOE spread, so that PSC-induced SOE dispersion remains bounded. Finally, because the PSC allocates power using a priority order, the resulting composite feasible set is conservatively tightened (via constraint \eqref{eqn:Thrm_new_cut_plane} and the SOE buffer \eqref{eqn:Thrm_energy}, which can reduce usable flexibility at coarse control time steps. However, this conservatism decreases with finer control time steps $\delta t$ (smaller SOE buffer $\tilde{\varepsilon}$) and with larger battery populations, for which $\tfrac{P^i_\text{max}}{\sum_i P^i_\text{max}} \rightarrow 0$.
\end{remark}


Next, we validate the proposed model through numerical case studies, examining scalability with fleet size and robustness to relaxing the assumptions made.

\section{Numerical Results}

This section evaluates the proposed realizable composite battery (RCB) model through numerical experiments focused on feasibility, solution optimality, and solve time. A unit-commitment (UC) problem is used as a practical example of a setting in which aggregate battery models are  integrated with system-level optimization. We consider a unit-commitment-based battery optimization use case from~\cite{arroyo_ensuring_2022}. The RCB model is compared against a full micro-model benchmark, in which each battery element is modeled individually using a mixed-integer linear program (MILP), and against a relaxed micro-model that does not explicitly enforce charging complementarity constraint. We then examine how these formulations perform as the fleet size increases, highlighting the scalability advantages of the composite approach. Finally, we relax the uniform power-to-energy ratio assumption and assess the resulting tradeoffs between feasibility, performance, and computational effort. 

\subsection{Unit Commitment Use Case}\label{subsec:UC_res_yes_assumptions}
The unit-commitment example introduced in~\cite{arroyo_ensuring_2022} considers a small system composed of two thermal generators and a battery energy storage system scheduled over two time periods. The generators are subject to binary on/off decisions, generation limits, and minimum up-time constraints, while the battery is operated using charging and discharging power variables and energy balance constraints. Following~\cite{arroyo_ensuring_2022}, linear generation costs and quadratic battery net-charging costs are considered, and system-wide power balance is enforced at each time step. The complete mathematical formulation of the unit commitment problem is provided in the Appendix~\ref{appx_UCproblem}. 

The battery energy storage system is modeled as an aggregation of $N$ smaller battery units. For each fleet size, the total battery energy capacity is set to 53.25~MWh, matching~\cite{arroyo_ensuring_2022}, and is distributed across the $N$ elements. Individual element energy limits are drawn around a mean value of $53.25/N$ with a $\pm20\%$ variation. The resulting values are then rescaled so that the total energy capacity of the fleet remains constant. A constant power-to-energy ratio of 0.5 is used for all elements.  Charging and discharging efficiencies $\eta_c$ and $\eta_d$ are held constant across the population.


To compare realizability, solution quality, and scalability, the unit-commitment problem is solved using three storage modeling approaches: (i) a full micro-model MILP with explicit complementarity constraint enforced at the element level, (ii) a relaxed micro-model that omits the complementarity constraint, and (iii) the proposed RCB model in~\eqref{eqn:Composite_Model_Theorem} coupled with the priority stack controller. The resulting solutions are implemented at the element-level to check for any constraint violations. This comparison is repeated for fleet sizes $N\in\{10,100,1000\}$. For the RCB model, the control time step is varied as $\delta t\in\{10~\mathrm{min},5~\mathrm{min},1~\mathrm{min}\}$ to study the effect of control granularity on solution quality. The results for the different models are shown in Table~\ref{tab:SimRes}. 
\begin{table}[h]
    \caption{Unit-commitment results}
    \begin{center}
    \setlength\tabcolsep{8pt}
     \centering
     \label{tab:SimRes}
    \begin{tabular}{lrrr}
    \toprule
    \textbf{Model}    & \textbf{Solve Time (s)} & \textbf{Cost (\$)} & \textbf{Realizable}\\ \midrule
    \noalign{\vskip -\aboverulesep}
    \rowcolor{gray!20}
        \multicolumn{4}{c}{\textbf{N=10}} \\
        \noalign{\vskip -\belowrulesep}
        \midrule
    MILP & 2.01 & 152.5  & \checkmark \\ 
    Relaxed & 1.95 & 152.5 & $\times$ \\ 
        \midrule
    RCB ($\delta t=$ 10 min) & 2.23  &  166.4 & \checkmark  \\ 
    RCB ($\delta t=$ 5 min) & 2.20 & 154.1  &  \checkmark\\ 
    RCB ($\delta t=$ 1 min) & 2.27 & 152.5 & \checkmark \\
    \midrule
    \noalign{\vskip -\aboverulesep}
    \rowcolor{gray!20}
        \multicolumn{4}{c}{\textbf{N=100}} \\
        \noalign{\vskip -\belowrulesep}
        \midrule
    MILP & 2.01 & 152.5 & \checkmark \\ 
    Relaxed & 2.05 & 152.5  & $\times$  \\ 
        \midrule
    RCB ($\delta t=$ 10 min)& 2.19   &  166.4 & \checkmark \\ 
    RCB ($\delta t=$ 5 min) & 2.24 & 154.1 & \checkmark \\ 
    RCB ($\delta t=$ 1 min)  & 2.26 & 152.5 & \checkmark \\ 
    \midrule
    \noalign{\vskip -\aboverulesep}
    \rowcolor{gray!20}
        \multicolumn{4}{c}{\textbf{N=1000}} \\
    \noalign{\vskip -\belowrulesep}
        \midrule
    MILP & 51.8  & 152.5 & \checkmark \\ 
    Relaxed & 49.7 & 152.5 &  $\times$\\ 
        \midrule
    RCB ($\delta t=$ 10 min)& 1.96 & 166.4  & \checkmark \\ 
    RCB ($\delta t=$ 5 min) & 2.81 & 154.1 &\checkmark  \\ 
    RCB ($\delta t=$ 1 min)  &  1.96& 152.5 &\checkmark  \\ 
    \bottomrule
    \end{tabular}
    \end{center}
    \end{table}
   

Table~\ref{tab:SimRes} highlights three key points. First, the relaxed micro-model attains the same objective value as the MILP but produces non-realizable schedules, confirming that ignoring element-level complementarity can lead to infeasible dispatch. Second, the proposed RCB formulation is realizable in all cases and exhibits solve times that are effectively independent of fleet size, in contrast to the rapidly increasing solve time of the MILP as $N$ grows. Third, for coarse control granularity ($\delta t=10$ min), the RCB model is conservative and yields higher cost, as the tightened realizability constraints \eqref{eqn:Thrm_new_cut_plane}--\eqref{eqn:Thrm_energy} limit usable battery flexibility. As $\delta t$ decreases, this conservatism diminishes and the RCB objective value converges to the MILP benchmark across all fleet sizes while maintaining near-constant solve time.


The 2-hour UC instance considered above is small enough that the MILP can still be solved within reasonable time despite being orders of magnitude slower than the RCB. However, to further assess scalability under a more realistic horizon, we also solved the UC problem over a 24-hour horizon ($K=24$) while varying the fleet size. Fig.~\ref{fig:UC_24h_solve_time} reports solve time versus $N$ for the MILP and RCB formulations. The MILP solve time increases sharply with the fleet size. In contrast, the RCB solve time remains nearly constant and comparable to the 2-hour case highlighting the computational advantage of the composite formulation.

\begin{figure}[h]
    \centering
    \includegraphics[width=0.9\linewidth]{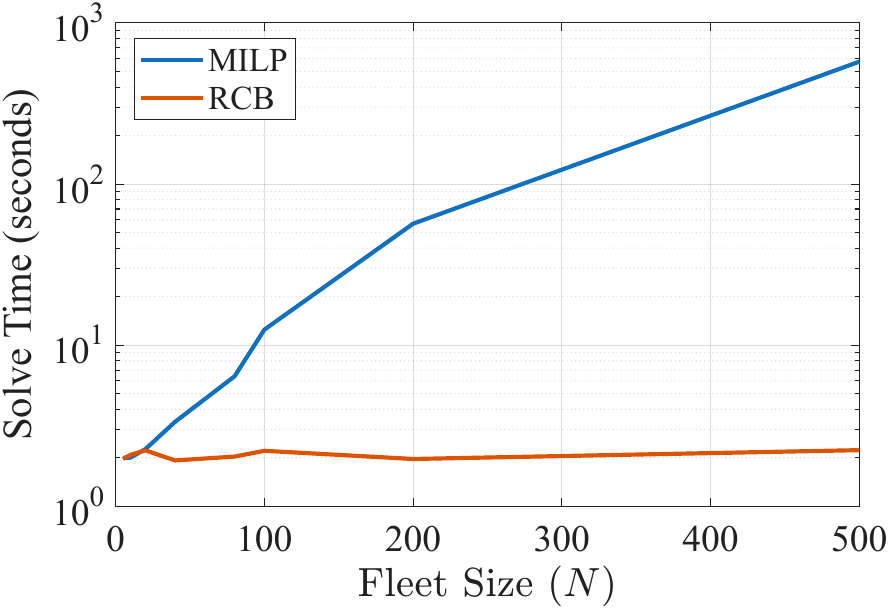}
    \caption{Solve time versus fleet size for the 24-hour UC problem.}
    \label{fig:UC_24h_solve_time}
\end{figure}


The results in Section~\ref{subsec:UC_res_yes_assumptions} are obtained under the modeling assumptions introduced in Section~\ref{subsec:assumptions}. In particular, the uniform power-to-energy ratio assumption is often the most limiting in heterogeneous fleets. For this reason, we next focus on relaxing Assumption~2 and quantifying the resulting tradeoff between realizability conservatism and performance.

\subsection{Relaxing Assumption 2: Power-to-Energy Ratios} \label{sec:relaxing_P_E_ratios}

The realizability conditions in Section~III rely on a uniform power-to-energy ratio across battery elements. When this condition is not satisfied, it can be artificially enforced by reducing either the $P^i_\text{max}$ or $E^i_\text{max}$ of each element to maintain a common ratio $P^i_\text{max}/E^i_\text{max}=\alpha$, as illustrated in Fig.~\ref{fig:PE_limiting}.
In the numerical experiments, the individual power-to-energy ratios are sampled uniformly from the interval $[0.25,1]$, corresponding to 4-hour and 1-hour batteries, respectively. The limiting procedure described above is then applied to enforce the uniform ratio required by the realizability conditions.
\begin{figure}
    \centering
    \includegraphics[width=0.9\linewidth]{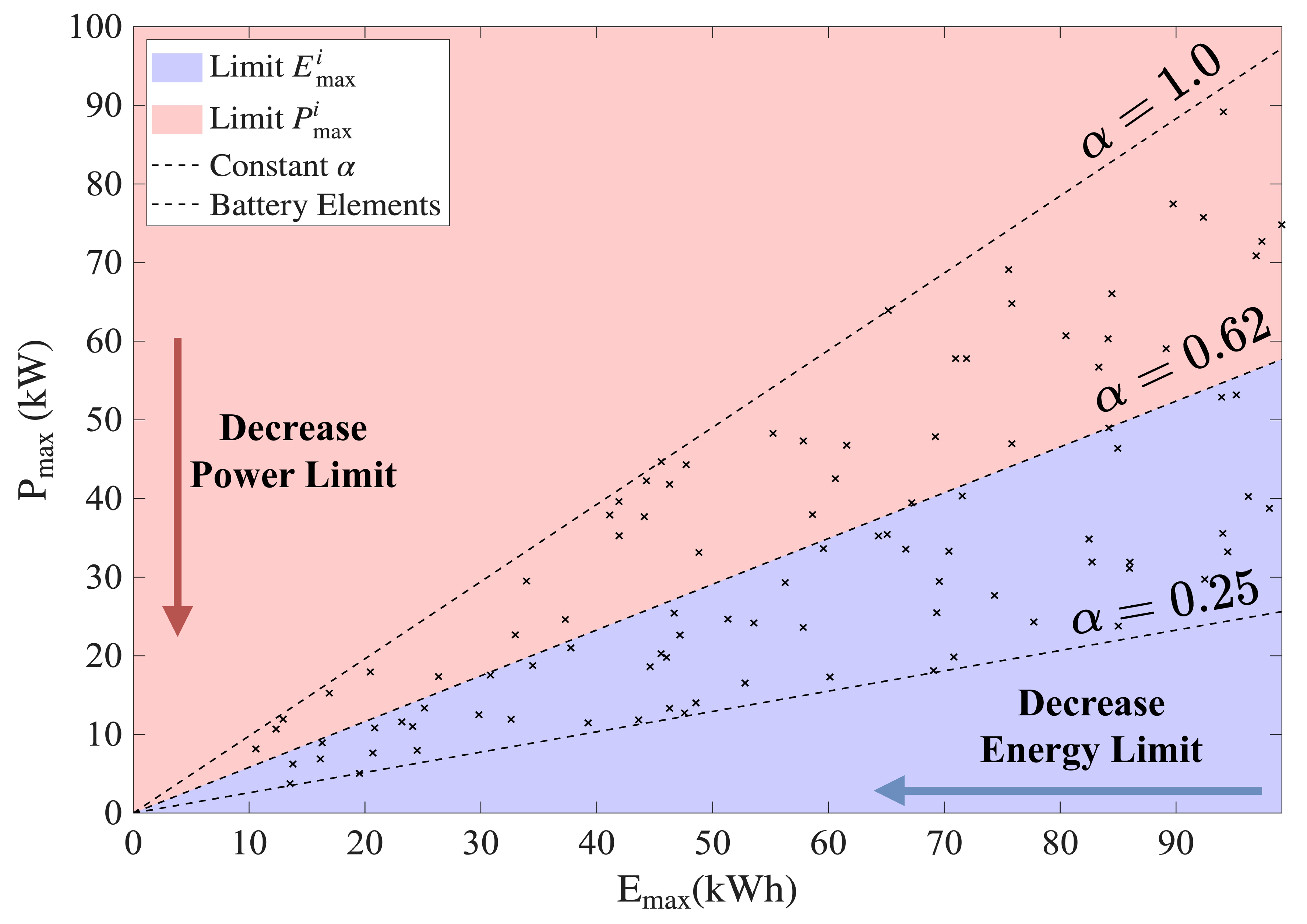}
    \caption{Illustration of enforcing a uniform power-to-energy ratio by limiting either the maximum power or usable energy of individual battery elements to match a target ratio $\alpha$.}
    \label{fig:PE_limiting}
\end{figure}

Imposing a uniform power-to-energy ratio through this limiting procedure  reduces the available power and/or energy of the battery fleet. The amount of usable power and energy depends on the chosen value of $\alpha$, leading to different levels of conservatism in the composite model. To quantify this effect, we sweep $\alpha$ over a range of values and evaluate the resulting impact on performance. Specifically, we parameterize the sweep using a scalar $\theta \in [0,1]$ and define 
\begin{equation}
\alpha(\theta) = \theta\,\max_i\!\left\{\tfrac{P^i_{\max}}{E^i_{\max}}\right\} + (1-\theta)\,\min_i\!\left\{\tfrac{P^i_{\max}}{E^i_{\max}}\right\},
\end{equation}
where $\theta=1$ enforces the largest power-to-energy ratio in the population and $\theta=0$ enforces the smallest, thereby sweeping between the least and most conservative limiting choices. Fig.~\ref{fig:alpha_sweep} shows the resulting effect on performance as $\theta$ is varied.

\begin{figure}
    \centering
    \includegraphics[width=0.9\linewidth]{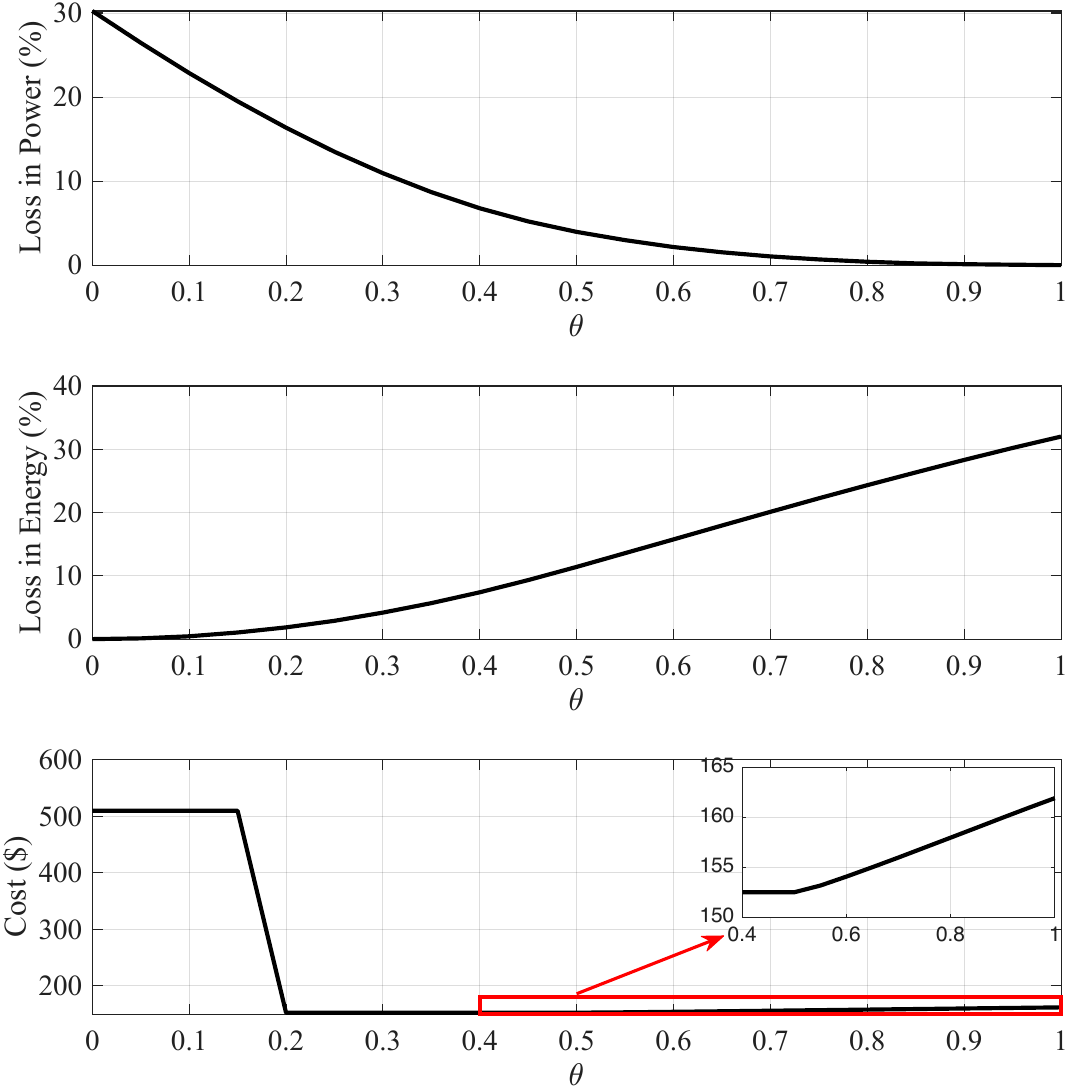}
    \caption{Effect of enforcing a uniform power-to-energy ratio on unit-commitment objective value as the parameter $\theta$ is varied from the largest to the smallest ratio in the battery population.}
    \label{fig:alpha_sweep}
\end{figure}

For large $\theta$, usable energy is significantly reduced while the impact on power capacity is minimal, and vice versa for small $\theta$. At small $\theta$, the loss of usable energy forces a higher cost generator to come online, leading to a sharp cost increase. For the UC use case considered, reductions in energy capacity have a larger impact on total cost than reductions in power capacity. The sweep further shows that the minimum cost is achieved for $\theta\in[0.25,0.55]$, beyond which the cost increases slightly as $\theta$ is increased.

\section{Conclusion}
This paper developed a realizable aggregate modeling framework for heterogeneous battery energy storage fleets while guaranteeing feasibility at the element level. We derived simple linear conditions under which aggregate charging and discharging schedules can be safely disaggregated while respecting element power limits, energy limits, and complementarity constraints.
The resulting RCB model yields realizable dispatch with solve times that are effectively independent of fleet size, and converges to the optimal solution (MILP solution) as control granularity is refined. Additional experiments examined cases with non-uniform power-to-energy ratios and showed how enforcing realizability through parameter limiting restores feasibility, with different limiting choices leading to different levels of conservatism and performance. Future work will focus on relaxing the remaining structural assumptions, extending the framework to market participation and network-constrained settings, and developing stochastic and distributionally robust composite battery formulations to explicitly account for uncertainty in battery parameters.

\section{AI Usage Disclosure}
AI tools were used to assist with improving writing clarity, spelling and grammar checks, and to help identify potentially relevant literature. All technical content, results, and interpretations are the authors’ own, and all cited sources were independently reviewed and selected by the authors.

\clearpage
\appendix 
\subsection{Proof of Proposition~\ref{prop: cutting plane}}\label{app:prop_1_proof}
\begin{proof}
Assume that $P_\mathrm{c}[k]$ and $P_\mathrm{d}[k]$ satisfy~\eqref{eqn:new cut plane}
and that $N_\mathrm{c}(l)+N_\mathrm{d}(l)>N$. Let $C$ and $D$ denote the sets of
elements assigned charging and discharging power by the PSC at control time step
$l$, respectively. Then $|C|=N_\mathrm{c}(l)$, $|D|=N_\mathrm{d}(l)$, and
$N_\mathrm{c}(l)+N_\mathrm{d}(l)>N$ implies $C\cap D\neq \emptyset$.

Under the PSC (Algorithm~\ref{alg:PSC}), all but at most one charging element and
all but at most one discharging element are fully saturated. Let $i_\mathrm{c}\in C$
and $i_\mathrm{d}\in D$ denote the (possibly) marginal charging and discharging
elements, respectively, with $P_\mathrm{c}^{i_\mathrm{c}}(l)>0$ and
$P_\mathrm{d}^{i_\mathrm{d}}(l)>0$. Then, by construction of the PSC,
\begin{align}\label{eq:PcPd_strict}
    P_\mathrm{c}[k] &> \sum_{i\in C\setminus\{i_\mathrm{c}\}} P^i_{\max},
    &
    P_\mathrm{d}[k] &> \sum_{i\in D\setminus\{i_\mathrm{d}\}} P^i_{\max}.
\end{align}

Adding the inequalities in \eqref{eq:PcPd_strict} gives
\begin{align}
    P_\mathrm{c}[k]+P_\mathrm{d}[k]
    &>
    \sum_{i\in C\setminus\{i_\mathrm{c}\}} P^i_{\max}
    +\sum_{i\in D\setminus\{i_\mathrm{d}\}} P^i_{\max} \notag\\
    &=
    \sum_{i=1}^{N} P^i_{\max}
    + \sum_{i\in C\cap D} P^i_{\max}
    - P^{i_\mathrm{c}}_{\max} - P^{i_\mathrm{d}}_{\max}. \label{eq:set_expand}
\end{align}

By design of the PSC, any overlap must satisfy  $i_\text{c},i_\text{d}\in C\cap D$. In the case where $i_\text{c}$ and $i_\text{d}$ are distinct, using  $P^{i_\mathrm{c}}_{\max},P^{i_\mathrm{d}}_{\max}\le \max_i\{P^i_{\max}\}$  we obtain
\begin{align}
    P_\mathrm{c}[k]+P_\mathrm{d}[k]
    &>
    \sum_{i=1}^{N} P^i_{\max}, \label{eq:union_minus_max}
\end{align}
which contradicts \eqref{eqn:new cut plane}. Furthermore, in the case where $i_\text{c}$ and $i_\text{d}$ are the same, we obtain
\begin{align}
    P_\mathrm{c}[k]+P_\mathrm{d}[k]
    &>
    \sum_{i=1}^{N} P^i_{\max} - P^{i_\mathrm{c}}_{\max}  \ge \sum_{i=1}^{N} P^i_{\max} - \max_i\{P^{i}_{\max}\}  , \label{eq:union_minus_max2}
\end{align}
which also contradicts \eqref{eqn:new cut plane}. In both cases we reach a contradiction; therefore $N_\mathrm{c}(l)+N_\mathrm{d}(l)\le N$,
and complementarity is satisfied. 
\end{proof}

\subsection{Proof of Proposition~\ref{prop: SOE difference bound}}\label{app:proof_prop2}

\begin{proof}
Consider any two elements $i$ and $j$ and, without loss of generality, assume $\tilde{E}^i(l)\ge \tilde{E}^j(l)$. Using the element-level SOE dynamics \eqref{eqn:Elements SOC}, we obtain
\begin{align}
\tilde{E}^i(l+1)-\tilde{E}^j(l+1) &= \tilde{E}^i(l)-\tilde{E}^j(l) \notag\\
+ \delta t\!\left(\frac{\eta_\text{c}^i P_\text{c}^i(l)}{E_\text{max}^i}
- \frac{\eta_\text{c}^j P_\text{c}^j(l)}{E_\text{max}^j}\right)
&- \delta t\!\left(\frac{P_\text{d}^i(l)}{\eta_\text{d}^i E_\text{max}^i}
- \frac{P_\text{d}^j(l)}{\eta_\text{d}^j E_\text{max}^j}\right).
\end{align}

Using Property~\ref{property:PSC} and Assumptions~1--3, the above difference is
bounded as
\begin{subequations}\label{eqn:SOC_diff_bound}
\begin{align}
    \tilde{E}^i(l+1)-\tilde{E}^j(l+1)&\le \tilde{E}^i(l)-\tilde{E}^j(l) \le \Delta \tilde{E}(l), \\
    \tilde{E}^i(l+1)-\tilde{E}^j(l+1) &\ge -\delta t\left(\eta_\text{c}\alpha_\text{c} + \frac{\alpha_\text{d}}{\eta_\text{d}}\right) = -\tilde{\varepsilon},\label{eqn:LastIneq}
\end{align}
\end{subequations}
where the inequality in~\eqref{eqn:LastIneq} uses $\tilde{E}^i(l)- \tilde{E}^j(l) \ge 0$. Taking the maximum over all $i,j$ yields $\Delta \tilde{E}(l+1) \leq  \max\left\{ \Delta \tilde{E}(l), \tilde{\varepsilon} \right\}$.
\end{proof}
\subsection{Proof of Proposition~\ref{prop:SOC new limits}}\label{app:proof_prop3}

\begin{proof}
From \eqref{eqn:Norm SOE order} and Proposition~\ref{prop: SOE difference bound}, we obtain
\begin{align}
    \tilde{E}^{\max}(l)-\tilde{E}(l) &\le \tilde{E}^{\max}(l)-\tilde{E}^{\min}(l) \le \tilde{\varepsilon}, \\
    \tilde{E}(l)-\tilde{E}^{\min}(l) &\le \tilde{E}^{\max}(l)-\tilde{E}^{\min}(l) \le \tilde{\varepsilon}.
\end{align}
Therefore, $\tilde{E}(l)\le 1-\tilde{\varepsilon}$ implies $\tilde{E}^{\max}(l)\le 1$, and $\tilde{E}(l)\ge \tilde{\varepsilon}$ implies $\tilde{E}^{\min}(l)\ge 0$. Since these implications hold for all control steps $l$, they therefore also hold for all scheduling timesteps, i.e., for all $k$, completing the proof.
\end{proof}

\subsection{Proof of Theorem~\ref{thm:heterogeneous_composite}}\label{appx_Thrm_proof}
\begin{proof}
Equations~\eqref{eqn:Elements SOC} and \eqref{eqn:SOE mapping} imply the composite dynamics~\eqref{eqn:Thrm_Composite SOE Update}, while \eqref{eqn:Elements SOC IC} and \eqref{eqn:SOE mapping} imply~\eqref{eqn:Thrm_Composite SOE IC}. By design, the PSC never exceeds individual power limits; therefore the aggregate bounds in~\eqref{eqn:Thrm_Pc_limit}--\eqref{eqn:Thrm_Pd_limit} ensure \eqref{eqn:Elements Pc Limit} and \eqref{eqn:Elements Pd Limit} are satisfied. From Proposition~\ref{prop: cutting plane}, the constraint \eqref{eqn:Thrm_new_cut_plane} ensures satisfaction of the element-level complementarity constraint \eqref{eqn:Elements CC} under the PSC. Similarly, if $\max_{i,j}\{|\tilde{E}^i_0-\tilde{E}^j_0|\}\le \tilde{\varepsilon}$, Proposition~\ref{prop:SOC new limits} guarantees that the element energy limits~\eqref{eqn:Elements SOC Limit} are respected whenever \eqref{eqn:Thrm_energy} holds. Therefore, any feasible composite dispatch is realizable under the PSC.
\end{proof}

\subsection{Unit Commitment Problem Formulation}\label{appx_UCproblem}
The unit commitment use case is adapted from~\cite{arroyo_ensuring_2022}. The objective is to minimize the total generation and battery net-charging costs over the scheduling horizon $\mathcal{K} = \{0, \dots, K-1\}$. The problem determines the commitment status $u_{g}[k] \in \{0,1\}$ and dispatch $P_{g}[k]$ for each generator $g \in \mathcal{G}$, as well as the composite battery power and energy variables defined in Section~\ref{subsec:composite_theorem}.
The optimization formulation is given by:

\begin{align}
    \min \sum_{k \in \mathcal{K}} \Delta t \Bigg[ &\sum_{g \in \mathcal{G}} C_{g} P_{g}[k]     + C_{\text{bat}} \left(P_\mathrm{c}[k] - P_\mathrm{d}[k]\right)^2 \Bigg] \label{eqn:UC_Obj}
\end{align}
subject to:
\begin{subequations}
\allowdisplaybreaks
\begin{align}
    & \sum_{g \in \mathcal{G}} P_{g}[k] + P_\mathrm{d}[k] = D[k] + P_\mathrm{c}[k], \quad \forall k \in \mathcal{K} \label{eqn:UC_Balance} \\
    & P_{g}^{\min} u_{g}[k] \le P_{g}[k] \le P_{g}^{\max} u_{g}[k], \quad \forall g, k \label{eqn:UC_GenLimits} \\
    & \sum_{k=0}^{L_g-1} (1-u_{g}[k]) = 0, \quad \forall g : L_g > 0 \label{eqn:UC_InitUp} \\
    & \sum_{\tau=k}^{k+UT_g-1} u_{g}[\tau] \ge UT_g (u_{g}[k] - u_{g}[k-1]), \nonumber \\
    & \quad \forall g, \forall k = L_g, \dots, K - UT_g \label{eqn:UC_MinUp} \\
    & \sum_{\tau=k}^{K-1} (u_{g}[\tau] - (u_{g}[k] - u_{g}[k-1])) \ge 0, \nonumber \\
    & \quad \forall g, \forall k = K - UT_g + 1, \dots, K-1 \label{eqn:UC_EndUp} \\
    & \text{Constraints \eqref{eqn:Thrm_Composite SOE Update}--\eqref{eqn:Thrm_energy}}. \label{eqn:UC_BatConstraints}
\end{align}
\end{subequations}

The objective function \eqref{eqn:UC_Obj} includes linear generation costs $C_g$ and quadratic penalty costs for battery net-charging $C_{\text{bat}}$, as defined in~\cite{arroyo_ensuring_2022}. Constraint \eqref{eqn:UC_Balance} enforces system power balance, where $D[k]$ is the system demand. Constraints \eqref{eqn:UC_GenLimits} enforce generator minimum and maximum output limits bounded by the binary commitment status $u_{g}[k]$. Constraints \eqref{eqn:UC_InitUp}--\eqref{eqn:UC_EndUp} enforce minimum up-time requirements $UT_g$, where $L_g$ represents the number of initial periods the generator must remain on due to prior history. Finally, \eqref{eqn:UC_BatConstraints} incorporates the realizable composite battery constraints derived in this paper.

The parameters used for the unit commitment case study are summarized in Table~\ref{tab:UC_Params}. The system consists of two thermal generators and one aggregate battery energy storage system.

\begin{table}[h]
\renewcommand{\arraystretch}{1.3}
\caption{Unit Commitment Problem Parameters}
\label{tab:UC_Params}
\centering
\begin{tabular}{l c l}
\hline
\textbf{Parameter} & \textbf{Value} & \textbf{Description} \\
\hline
$K$ & 2 h & Prediction horizon \\
$\Delta t$ & 1 h & Scheduling time step duration \\
$D[0]$,$D[1]$ & 45,80 MW & System demand profile \\
\hline
$P_{1}^{\max}$ & 60 MW & Max output for Gen 1 \\
$P_{2}^{\max}$ & 25 MW & Max output for Gen 2 \\
$P_{1}^{\min}$ & 50 MW & Min output for Gen 1 \\
$P_{2}^{\min}$ & 20 MW & Min output for Gen 2 \\
$C_{1}$ & \$1/MWh & Generation cost for Gen 1 \\
$C_{2}$ & \$20/MWh & Generation cost for Gen 2 \\
$UT_{1}$ & 2 h & Minimum up-time for Gen 1 \\
$UT_{2}$ & 0 h & Minimum up-time for Gen 2 \\
$L_{1}, L_{2}$ & 0 & Initial on-time history \\
\hline
$E_{\max}$ & 53.25 MWh & Total energy capacity \\
$P_{\max}$ & 25 MW & Rated power capacity \\
$\eta_{\text{c}}, \eta_{\text{d}}$ & 0.95 & Charging/Discharging efficiency \\
$E_{0}$ & 20 MWh & Initial stored energy \\
$C_{\text{bat}}$ & \$0.1/MW$^2$h & Net charging penalty coefficient \\
\hline
\end{tabular}
\end{table}

\bibliographystyle{ieeetr}
\bibliography{references}

\begin{thebibliography}{10}

\bibitem{FERC2222}
``{Order No. 2222: Participation of Distributed Energy Resource Aggregations in Markets Operated by Regional Transmission Organizations and Independent System Operators}.'' Fact sheet, Federal Energy Regulatory Commission, 2020.
\newblock Accessed: 2026-01-29.

\bibitem{MultiBatteryEV2023}
F.~Al~Taha, T.~Vincent, and E.~Bitar, ``A multi-battery model for the aggregate flexibility of heterogeneous electric vehicles,'' in {\em 2023 American Control Conference (ACC)}, pp.~1243--1250, 2023.

\bibitem{Muller2019AggDisagg}
F.~L. Müller, J.~Szabó, O.~Sundström, and J.~Lygeros, ``Aggregation and disaggregation of energetic flexibility from distributed energy resources,'' {\em IEEE Transactions on Smart Grid}, vol.~10, no.~2, pp.~1205--1214, 2019.

\bibitem{Wen2022ExactAFR}
Y.~Wen, Z.~Hu, S.~You, and X.~Duan, ``Aggregate feasible region of ders: Exact formulation and approximate models,'' {\em IEEE Transactions on Smart Grid}, vol.~13, no.~6, pp.~4405--4423, 2022.

\bibitem{Ozturk2024VertexBased}
E.~Öztürk, T.~Faulwasser, K.~Worthmann, M.~PreißInger, and K.~Rheinberger, ``Alleviating the curse of dimensionality in minkowski sum approximations of storage flexibility,'' {\em IEEE Transactions on Smart Grid}, vol.~15, no.~6, pp.~5733--5743, 2024.

\bibitem{Muller2015Zonotopes}
F.~L. Müller, O.~Sundström, J.~Szabó, and J.~Lygeros, ``Aggregation of energetic flexibility using zonotopes,'' in {\em 2015 54th IEEE Conference on Decision and Control (CDC)}, pp.~6564--6569, 2015.

\bibitem{Nazir2018UnionBasedMinkowski}
M.~S. Nazir, I.~A. Hiskens, A.~Bernstein, and E.~Dall'Anese, ``Inner approximation of minkowski sums: A union-based approach and applications to aggregated energy resources,'' in {\em 2018 IEEE Conference on Decision and Control (CDC)}, pp.~5708--5715, 2018.

\bibitem{prat_2024}
E.~Prat, {\em Market-based scheduling of energy storage systems: Optimality guarantees}.
\newblock PhD thesis, DTU - Technical University of Denmark, Lyngby, 11 2024.

\bibitem{RobustNawaf}
N.~Nazir and M.~Almassalkhi, ``Guaranteeing a physically realizable battery dispatch without charge-discharge complementarity constraints,'' {\em IEEE Transactions on Smart Grid}, vol.~14, no.~3, pp.~2473--2476, 2023.

\bibitem{MILP_for_CC}
J.~Hu, J.~E. Mitchell, J.-S. Pang, K.~P. Bennett, and G.~Kunapuli, ``On the global solution of linear programs with linear complementarity constraints,'' {\em SIAM Journal on Optimization}, vol.~19, no.~1, pp.~445--471, 2008.

\bibitem{Sufficient_Cond_for_Relax}
Z.~Li, Q.~Guo, H.~Sun, and J.~Wang, ``Sufficient conditions for exact relaxation of complementarity constraints for storage-concerned economic dispatch,'' {\em IEEE Transactions on Power Systems}, vol.~31, no.~2, pp.~1653--1654, 2016.

\bibitem{Lin_Exact_Cond_Ex_Ante}
W.~Lin, C.~Y. Chung, and C.~Zhao, ``Relaxing complementarity constraints of energy storage with feasibility and optimality guarantees,'' in {\em IEEE Power and Energy Society General Meeting}, pp.~1--5, 2023.

\bibitem{positive_prices}
P.~Haessig, ``Convex storage loss modeling for optimal energy management,'' in {\em IEEE Madrid PowerTech}, pp.~1--6, 2021.

\bibitem{Jakob_Pos_Quadratic_Cost}
D.~Wu, T.~Yang, A.~A. Stoorvogel, and J.~Stoustrup, ``Distributed optimal coordination for distributed energy resources in power systems,'' {\em IEEE Transactions on Automation Science and Engineering}, vol.~14, no.~2, pp.~414--424, 2017.

\bibitem{garifi_convex_2020}
K.~Garifi, K.~Baker, D.~Christensen, and B.~Touri, ``Convex {Relaxation} of {Grid}-{Connected} {Energy} {Storage} {System} {Models} {With} {Complementarity} {Constraints} in {DC} {OPF},'' {\em IEEE Transactions on Smart Grid}, vol.~11, pp.~4070--4079, Sept. 2020.

\bibitem{Composite_Homogeneous}
M.~Elsaadany, M.~R. Almassalkhi, and S.~H. Tindemans, ``Linear aggregate model for realizable dispatch of homogeneous energy storage,'' {\em IEEE Control Systems Letters}, vol.~9, pp.~1267--1272, 2025.

\bibitem{karan_kalsi}
S.~P. Nandanoori, S.~Kundu, D.~Vrabie, K.~Kalsi, and J.~Lian, ``Prioritized threshold allocation for distributed frequency response,'' in {\em 2018 IEEE Conference on Control Technology and Applications}, pp.~237--244, 2018.

\bibitem{equal_unequal}
M.~Bauer, M.~Muehlbauer, O.~Bohlen, M.~A. Danzer, and J.~Lygeros, ``Power flow in heterogeneous battery systems,'' {\em Journal of Energy Storage}, vol.~25, p.~100816, 2019.

\bibitem{equal_unequal2}
M.~Mühlbauer, O.~Bohlen, and M.~A. Danzer, ``Analysis of power flow control strategies in heterogeneous battery energy storage systems,'' {\em Journal of Energy Storage}, vol.~30, p.~101415, 2020.

\bibitem{tesla_2019}
{Tesla Inc.}, ``Tesla powerwall datasheet,'' Sept. 2022.
\newblock \url{https://tesla-cdn.thron.com/static/KBQ2AZ_Tesla_Powerwall__Datasheet_NA-EN_Y3THJT.pdf?xseo=&response-content-disposition=inline\%3Bfilename\%3D\%22powerwall-plus-datasheet-na-en.pdf\%22} (accessed Oct. 3, 2022).

\bibitem{arroyo_ensuring_2022}
J.~M. Arroyo, ``Ensuring {Physically} {Realizable} {Storage} {Operation} in the {Unit} {Commitment} {Problem},'' {\em IEEE Transactions on Power Systems}, vol.~37, pp.~4966--4969, Nov. 2022.

\end{thebibliography}

\endgroup
\end{document}